\documentclass[a4paper, 12pt]{article}

\usepackage[all]{xy}
\usepackage{amsfonts}
\usepackage{amssymb, amscd}
\usepackage{latexsym}

\usepackage{amsmath} 
\usepackage{amsthm}

\newtheorem{theorem}{Theorem}[section]

\newtheorem{prop}[theorem]{Proposition}
\newtheorem{cor}[theorem]{Corollary}

\newtheorem{remark}[theorem]{Remark}
\newtheorem{definition}[theorem]{Definition}

\newtheorem{defin}[theorem]{Definition}

\theoremstyle{definition}
\newtheorem{example}[theorem]{Example}

\begin{document}


\title{Logically automorphically equivalent knowledge bases}

\maketitle

\begin{center}

\author{

 E.~Aladova, T.~Plotkin

\smallskip
 {\small
              Bar Ilan University,

          5290002, Ramat Gan, Israel

               {\it E-mail address:} aladovael (at) mail.ru

                                    \hspace{80pt} plot (at) macs.biu.ac.il


            }
}
\end{center}

\begin{abstract}
Knowledge bases theory provide an important example of the field
where applications of universal algebra and algebraic logic look
very natural, and their interaction with practical problems
arising in computer science might be very productive.

In this paper we study the equivalence problem for knowledge
bases.  Our  interest is to find out how the informational
equivalence is related to the logical description of knowledge.
The main objective of this paper is logically automorphically
equivalent knowledge bases. As we will see this notion gives us a
good enough characterization of knowledge bases.
\end{abstract}



\section{Motivation}




Research work on databases was began at the end of the  sixties.
The classical work of Codd~\cite{Codd} gave the theoretical
foundation of databases. From this time  advances in database
theory  are closely related  to mathematical logic, theory of
algorithms and general algebra.

Knowledge base systems go far beyond the relational database
model. They require complex data processing, which may include
rules. Knowledge base systems combine database features and
artificial intelligence techniques.

Investigations in  database theory  led to the  formal definitions
for various types of databases, 
whereas knowledge bases are often defined informally.
There are different views on knowledge bases which are discussed,
in particular,  in the information technology, strategic
management, and organizational theory literature.
The informal representation of a knowledge base, for example, does
not allow identifying the duplicate information represented in
different formats by different knowledge base implementations. The
formal mathematical model of a knowledge base allows to get formal
solutions for various problems arise in knowledge bases theory, in
particular, for equivalence problem for knowledge bases.

Plotkin in~\cite{P-DB} proposed a mathematical model of a database
and gave a formal definition of the databases equivalence concept.
Researches in this area give rise to the  algebraic model of a
knowledge base which was  introduced and developed  in
\cite{Plotkin_AG}, \cite{Seven}, \cite{P-DB}.  The main
peculiarity of this approach is that a database and a knowledge
base is considered as a certain algebraic structure. The
mathematical model of a knowledge base (or database), which is
viewed  as an algebraic structure helps us to understand the
nature of a concrete real knowledge base (database), and it
enables to solve various problems in knowledge bases (databases)
theory.
The model of a knowledge base involves various ideas of universal
algebra, algebraic logic and  algebraic geometry.
Such a model is useful for many reasons.
There are a lot of specialized knowledge bases and it is desirable
to determine their characteristic properties without referring to
their complicated structure and without studying their detailed
architecture. Mathematical model of a knowledge base allows  to
distinguish some invariants
of knowledge bases which rigidly determine them.

In this paper we discuss the equivalence problem for knowledge
bases.
%
This problem goes back to the similar one for databases. It was
first posed by Aho, Sagiv, Ullman in~\cite{ASU} and  Beeri,
Mendelzon, Sagiv, Ullman in~ \cite{BeMendSagUll} and gave rise to
the notion of databases schemes equivalence. They propose an
approach to databases schemes equivalence based on the notion of a
fixed point. In this setting two relational database schemes are
equivalent if their sets of fixed points coincide.
 Correspondingly, two relational databases are equivalent if their sets of all fixed points intersected
 with the sets of feasible instances coincide. This and other approaches  to the database equivalence problem
 had been studied in  numerous papers (see \cite{AABM}, \cite{BM}, \cite{Ba}, \cite{Beniaminov},
 \cite{Heu}, \cite{Pt_1}, \cite{Ri},
 etc.).

We are interested in a special kind of equivalence, namely,
\emph{informational equivalence}.
%
%
Informally, one can say that  two knowledge bases are
informationally equivalent if and only if all information that can
be retrieved from the one knowledge base can be also obtained from
the other one and vice versa.
The formal mathematical model of a knowledge base allows  to solve
formally the informational equivalence problem. Various solutions
for the knowledge bases equivalence problem based on algebraic
geometry approach were obtained in \cite{KnPlot-KB-inf-equiv},
\cite{Plotkin-AlgCatDB}, \cite{Seven}, \cite{P-DB}, \cite{PP_AA},
\cite{PP_AU}, \cite{PlKn-Verif}.
This paper continues the research of the knowledge bases
equivalence problem based on logical geometry approach which was
started in \cite{APP}.

The paper is organized as follows. In Section~\ref{sec:Prelim} we
give a brief review of basic notions and notations from universal
algebraic geometry. In particular, we define Halmos categories and
construct the Galois correspondence which is very important in our
considerations. The material of this section can be found in the
papers of B.~Plotkin (\cite{Plotkin_AG}, \cite{Pl_GAGTA},
\cite{Seven}, see also  \cite{BI-90-I} for some detailed proofs).
In Section~\ref{sec:KB-model} we introduce a knowledge base model
under consideration.
Section~\ref{sec:KB-equiv}   deals with 
various
equivalences of knowledge bases and connections between them. In
particular, we give the formal definition of  informationally
equivalent knowledge bases.
In Section~\ref{sec:LogAutEquiv-KB} we introduce one more
equivalence for knowledge bases, namely, \emph{logically
automorphical equivalence}, and  present the main result of the
paper which state that \emph{logically automorphically equivalent
knowledge bases are informationally equivalent}
(Theorem~\ref{th:Aut-Equiv-KB}).
%

\section{Preliminaries: Mathematical apparatus}\label{sec:Prelim}

\subsection{Basic notions and notations}\label{sec:BasicNotions}


Let $X^{0}=\{x_1,\dots ,x_n , \dots \}$ be an infinite set of
variables. Denote by $\Gamma$ the collection of all finite subsets
$X$ of $X^{0}$.

Let $\Theta$ be a \emph{variety of algebras}, that is a class of
algebras satisfying a set of identities (see, for instance, 
\cite{Malcev}, \cite{P-DB}). We denote by $Var(H)$ the variety
generated by the algebra $H$.

Denote by $W(X)$ the \emph{free algebra} in the variety $\Theta$
with free generating set $X$, $X\in \Gamma$. All free algebras
$W(X)\in \Theta$, form a \emph{category of free algebras}
$\Theta^0$ with homomorphisms $s : W(X) \to W(Y)$ as morphisms,
$X, Y\in\Gamma$.

By a \emph{model} $\cal H$ we mean a triple $(H, \Psi, f)$, where
$H$ is an algebra from $\Theta$, $\Psi$ is a set of relation
symbols $\varphi$,  $f$ is an interpretation of all $\varphi$ in
$H$ (see, for instance, \cite{ChangKeisler_ModelTh},
\cite{Marker}, \cite{P-DB}).

Take an algebra $H$ in $\Theta$. A \emph{point} $(a_1, \ldots,
a_n)$ from $n$-th Cartesian power of $H$ can be represented as a
map $\mu: X \to H$ such that $a_i=\mu(x_i)$. This map can be
extended up to homomorphism of algebras $\mu: W(X)\to H$.
Thus, the point $(a_1, \ldots, a_n)$ can be also viewed as a
homomorphism $\mu: W(X) \to H$.

Denote by $Hom(W(X),H)$ the set of all  homomorphism  from $W(X)$
to $H$. We will regard $Hom(W(X),H)$ as an \emph{affine space}.

All affine spaces $Hom(W(X),H)$ with various $X\in \Gamma$
constitute the \emph{category $\Theta^{0}(H)$ of  affine spaces}
with morphisms
$$
\widetilde s : Hom(W(X),H) \to Hom(W(Y),H),
$$
for each homomorphism of free algebras $s : W(Y) \to W(X)$. The
map $\widetilde s$ is defined as $\widetilde s(\mu)=\mu s$, where
$\mu: W(X) \to H$, $\widetilde s(\mu): W(Y) \to H$.

The categories $\Theta^{0}$ and $\Theta^{0}(H)$ are very important
for further considerations. Moreover, the following theorem takes
place.

\begin{theorem}[\cite{MPP-Lie}]
The categories $\Theta^{0}$ and $\Theta^{0}(H)$ are  dual if and
only if $Var(H) =\Theta$.
\end{theorem}

\subsection{Halmos categories}\label{sec:HalmosCat}


Halmos categories  were introduce in papers of B.I.~Plotkin
\cite{Plotkin_AG}, \cite{Seven}. 
Halmos categories are related to the first-order logic in a way
analogous to the relationship between boolean algebras and
propositional logic. Such an approach allows us to use technics
and structures of algebraic logic (see \cite{Halmos},
\cite{P-DB}). The immediate advantage of this phenomenon is that
we can view queries to a knowledge base and replies to these
queries as objects of the same nature, i.e., objects of Halmos
categories. Then the transition query-reply can be treated as a
functor (for details see Section~\ref{sec:KB-model}).

We start from the notion of an existential quantifier on a boolean
algebra. Let $B$ be a boolean algebra. Existential quantifier on
$B$ is a unary operation $\exists : B \to B$ such that the
following conditions hold:
\begin{enumerate}
\item $\exists \ 0 = 0$, \item $a \leq \exists a$, \item $\exists
(a \wedge \exists b) = \exists a \wedge \exists b$.
\end{enumerate}

Universal quantifier $\forall : B \to B$ is dual to $\exists : B
\to B$, they are related by  $ \forall a=\neg(\exists (\neg a))$.

\begin{definition}\label{Def:ExtendedBA}
Let a set of variables $X=\{x_1,\dots ,x_n\}$ and a set of
relations $\Psi$ be given. A boolean algebra $B$ is called an
extended boolean algebra over $W(X)$ relative to $\Psi$, if

1. the existential quantifier $\exists x$ is defined on $B$ for
all $x \in X$, and  $\exists x \exists y = \exists y \exists x$
for all $x,y \in X$;

2. to every  relation symbol $\varphi \in \Psi$ of arity
$n_\varphi$ and a collection of elements $w_1,\ldots,
w_{n_\varphi}$ from $W(X)$ there corresponds a nullary operation
(a constant) of the form $\varphi(w_1,\ldots, w_{n_\varphi})$ in
$B$.
\end{definition}

Thus, the signature $L_X$ of an extended boolean algebra consists
of the boolean connectives, existential quantifiers  $\exists x$
and constants $\varphi(w_1,\ldots, w_{n_\varphi})$:
$$
L_X = \{\vee, \wedge, \neg, \exists x, M_X \},
$$
where $M_X$ is a set of all $\varphi(w_1, \ldots, w_{n_\varphi})$.

There are two important examples of extended boolean algebras.

\begin{example}\label{ex:Bool}
Let a model ${\cal H}=(H,\Psi,f)$ be given. Let $Bool(W(X),H)$ be
the boolean algebra of all subsets in $Hom(W(X),H)$. One can equip
the boolean algebra $Bool(W(X),H)$ with the structure of an
extended boolean algebra (\cite{APP}, \cite{BI-90-I},
\cite{P-DB}). We denote this extended boolean algebra by
$Hal_{\Theta}^{X}({\cal H})$.
\end{example}

\begin{example}\label{ex:Phi}
Another important example of an extended boolean algebra is
presented by the algebra of formulas $\Phi_{0}(X)=\mathfrak
L_X/\tau_X$, where $\mathfrak L_X$ is the absolutely free algebra
in the signature $L_X$ over the set $M_X$, $\tau_X$ is a
congruence relation on $\mathfrak L_X$ defined by the rule: $
u\tau_X v \mbox{ if and only if } \vdash (u\to v)\wedge (v\to u),
$ $u, v \in \mathfrak L_X$. Boolean operations and quantifiers on
$\Phi_{0}(X)$ are naturally inherited from $\mathfrak L_X$. For
more details see \cite{APP}, \cite{BI-90-I}, \cite{Plotkin_AG},
\cite{P-DB}, \cite{PAP}.
\end{example}

Now we define \emph{a Halmos category} which plays a very
important role in further considerations.

\begin{definition}
A category $\mathbb H$ is a Halmos category if:
\begin{enumerate}
\item
 Every its object has the form $\mathbb H(X)$, where $\mathbb H(X)$ is an extended
boolean algebra over $W(X)$.

\item
 The morphisms in $\mathbb H$ correspond to morphisms in the category $\Theta^{0}$.
To every morphism  $s: W(X) \to W(Y)$ in $\Theta^0$ it corresponds
a morphism $s_{*}: \mathbb H(X) \to \mathbb H(Y)$ in $\mathbb H$
such a way that
\begin{enumerate}
\item %
the transitions $W(X)\to \mathbb H(X)$ and $s \to s_{*}$ determine
a covariant functor from $\Theta^{0}$ to $\mathbb H$.

\item %
$s_*: \mathbb H(X) \to \mathbb H(Y)$ is a homomorphism of
corresponding boolean algebras.
\end{enumerate}

\item %
There are special identities controlling the interaction of
morphisms with quantifiers and constant (for details see
\cite{Plotkin_AG}, \cite{PAP}, \cite{PP_D}).
\end{enumerate}
\end{definition}

Next two examples of Halmos categories are based on examples of
extended boolean algebras above.

\begin{example}\label{ex:CategoryHal}
{\bf Category $Hal_{\Theta}(\cal H)$.} Objects of this category
are extended boolean algebras $Hal_{\Theta}({\cal H})$ from
Example~\ref{ex:Bool} for various $X\in \Gamma$. Morphisms
$$
s_{*}: Hal^{X}_{\Theta}({\cal H})\to Hal^{Y}_{\Theta}({\cal H}),
$$
 are defined as follows:
$$
\mu\in s_{*}A  \Leftrightarrow \ \mu s\in A,
$$
where $\mu:W(Y)\to H$, $A\subset Hom(W(X),H)$, $s:W(X)\to W(Y)$.

\begin{remark}\label{rem:s*A}
A homomorphism $s:W(X)\to W(Y)$ gives rise a map
$$
\widetilde s: Hom(W(Y),H) \to Hom(W(X),H).
$$
by the rule $\widetilde s (\mu)=\mu s$, which is a morphism in the
category of affine spaces $\Theta^{0}(H)$ (see
Section~\ref{sec:BasicNotions}). Let a subset $A$ from $Hom(W(X),
H)$ be given. Then $s_{*}A$ is the full pre-image of $A$ under
$\widetilde s$. 
\end{remark}

\end{example}

\begin{example}\label{ex:CategoryPhi}
{\bf Category $\widetilde \Phi$.}
Objects of the category $\widetilde \Phi$ are constructed using
extended boolean algebras $\Phi_{0}(X)$, $X\in \Gamma$. Denote by
$[\varphi (w_1, \ldots, w_{n_\varphi})]_{\tau_X}$ the image of the
element $\varphi(w_1, \ldots, w_{n_\varphi})\in M_X$ under the
homomorphism $\mathfrak L_X\to \Phi_{0}(X)=\mathfrak L_X/\tau_X$.
Let $[M_{X}]_{\tau_{X}}$ be a set of all $[\varphi (w_1, \ldots,
w_{n_\varphi})]_{\tau_X}$, $\varphi\in \Psi$.

A homomorphism of free algebras  $s:W(X)\to W(Y)$ induces the map
$ s_{*}: [M_{X}]_{\tau_{X}} \to [M_{Y}]_{\tau_{Y}},$ by the rule
$$
s_{*}\big([\varphi (w_1, \ldots, w_{n_\varphi})]_{\tau_X}\big)=
[\varphi (sw_1, \ldots, sw_{n_\varphi})]_{\tau_Y}.
$$
 This map can
be extended up to homomorphism of boolean algebras
$s_{*}:\Phi_{0}(X)\to \Phi_{0}(Y)$.

Note that morphism of a Halmos category $s_{*}$ correlates with
quantifiers under the certain rules (see \cite{Plotkin_AG},
\cite{PAP}) and they are not homomorphisms of extended boolean
algebras. Thus, the extended boolean algebras $\Phi_{0}(X)$ cannot
be an object of the category $\widetilde \Phi$. We should to add
to each $\Phi_{0}(X)$ all formulas of the form $s_{*}u$, $u\in
\Phi_{0}(X)$. Denote objects of the category $\widetilde \Phi$ by
$\Phi(X)$. So, the category $\widetilde \Phi$ is a category with
objects of the form $\Phi(X)$ and morphisms $s_{*}:\Phi(X)\to
\Phi(Y)$, $X, Y\in \Gamma$.
\end{example}


The next remark is connected with Remark~\ref{rem:s*A}.

\begin{remark}\label{rem:TildeS-T}
Let a homomorphism $s: W(X)\to W(Y)$ be given. In parallel to the
map $\widetilde s: Hom(W(Y), H)\to Hom(W(X),H)$, 
we define a map from the set of subsets in $\Phi(Y)$ to the set of
subsets in $\Phi(X)$. We denote it the same symbol $\widetilde s$
and define as
$$
{\widetilde s}T=\{ u\in \Phi(X) \mid s_{*}u\in T \},
$$
where $T\subset \Phi(Y)$. Then $\widetilde s T$ is the full
pre-image of $T$ under $s_{*}$.
\end{remark}



Halmos categories $\widetilde\Phi$ and $Hal_{\Theta}({\cal H})$
are tightly connected via homomorphism of extended boolean
algebras
$$
Val^X_{\cal H} : \Phi(X)\to Bool(W(X),H). 
$$
Intuitively, the image of a formula $u\in \Phi(X)$ under the
homomorphism $Val^X_{\cal H}$ is a value of $u$ in the algebra
$H$, i.e. $Val^X_{\cal H} u$ is a set of point in $Hom(W(X), H)$
satisfied $u$. For details see  \cite{Plotkin_AG}, \cite{Seven},
\cite{PAP}.

Let $\mu$ be a point from the affine space $Hom(W(X),H)$.

\begin{definition}
The logical kernel $LKer(\mu)$ of a point $\mu$ is the set of all
formulas $u\in\Phi(X)$ which hold true on the point $\mu$, that is
$$
LKer(\mu)=\{u\in\Phi(X) \mid  \mu\in Val^{X}_{\cal H}(u) \}
$$
\end{definition}

Note that the logical kernel $LKer(\mu)$ of a point $\mu$ is a
boolean ultrafilter (maximal filter) in the algebra $\Phi(X)$ (see
\cite{Plotkin_IsotAlg}).

\subsection{Galois correspondence}\label{sec:GaloisCor}

Now we define a correspondence between sets of formulas in the
algebra $\Phi(X)$ and subsets of points from the affine space
$Hom(W(X),H)$. 

Let $T$ be a set of formulas from $\Phi(X)$. We define a set of
points  $T^{L}_{\cal H}$ in $Hom(W(X),H)$ as
$$
 T^{L}_{\cal H}=\{ \mu :W(X)\to H \mid
T\subset LKer(\mu) \}.
$$
That is, $ T^{L}_{\cal H}$ is a set of all points $\mu\in
Hom(W(X),H)$ satisfying all formulas from $T\subset\Phi(X)$. The
set $T^{L}_{{\cal H}}$ can be written as follows:
$$
T^{L}_{\cal H}=\bigcap_{u\in T} Val^{X}_{\cal H}(u).
$$

Take a set of points $A\subset Hom(W(X,H))$ and define a set of
formulas $A^{L}_{\cal H}$ in $\Phi(X)$:
$$
A^{L}_{\cal H}=\{ u\in\Phi(X) \mid A\subset Val^{X}_{\cal H}(u)
\}.
$$
The set $ A^{L}_{\cal H}$ is the set of all formulas $u \in
\Phi(X)$ hold true at all points from $A$. One can present the set
$A^{L}_{\cal H}$ as follows:
$$
A^{L}_{\cal H}=\bigcap_{\mu\in A} LKer(\mu).
$$

The defined above correspondence between sets of formulas and sets
of points is the Galois correspondence (see \cite{MacLane}). In
the case of the Galois correspondence one can speak about Galois
closures. In particular, subsets $T^{L}_{\cal H}\subset
Hom(W(X),H)$ and $A^{L}_{\cal H}\subset \Phi(X)$ are
Galois-closed.

We call the  subset $T^{L}_{\cal H}\subset Hom(W(X),H)$
\emph{definable set} presented by the set of formulas $T$. The set
$A^{L}_{\cal H}$ is  a boolean filter in the algebra
$\Phi(X)$, as 
an intersection of boolean filters $LKer{\mu}$. It is called
\emph{$\cal H$-closed filter}.

The constructed above Galois correspondence give us a bijection
between definable sets in $Hom(W(X),H)$ and $\cal H$-closed
filters in the extended boolean algebra $\Phi(X)$.

It is known the following proposition.

\begin{prop}[\cite{Plotkin_IsotAlg}]\label{prop:InterH-closedFilt}
The intersection of $\cal H$-closed filters is an $\cal H$-closed
filter. $\qquad\square$
\end{prop}


The next proposition describes one more property of $\cal
H$-closed filters (\cite{BI-90-I}, \cite{Plotkin_AG}).

\begin{prop}\label{prop:tilde-sT-closedFilt}
Let a homomorphism of free algebras  $s: W(X)\to W(Y)$ be given.
If $T$ is an $\cal H$-closed filter in $\Phi(Y)$, then
${\widetilde s }T$ is an $\cal H$-closed filter in $\Phi(X)$.
\end{prop}

\begin{cor}\label{cor:sets+tilde-s}
Let $A$ be a set of points in $Hom(W(Y),H)$ and $s: W(X)\to W(Y)$
be a homomorphism of free algebras. Then
$$
\widetilde s (A^{L}_{\cal H})=\big( \widetilde s A \big)^{L}_{\cal
H}.
$$
\end{cor}


The next proposition describes the relation between the Galois
correspondence and morphisms in the categories $Hal_\Theta({\cal
H})$ and $\widetilde \Phi$ (see \cite{BI-90-I},
\cite{Plotkin_AG}).

\begin{prop}\label{prop:Gal+s*}
Let $T$ be a set of formulas from $\Phi(X)$, $A$ be a set of
points in $Hom(W(X),H)$ and $s: W(X)\to W(Y)$ be a homomorphism of
free algebras. Then
\begin{enumerate}
\item %
$(s_*T)^{L}_{{\cal H}}=s_{*} T^{L}_{\cal H}$,

\item %
$s_{*} A^{L}_{\cal H}\subseteq (s_*A)^{L}_{{\cal H}}$.

\end{enumerate}
\end{prop}

\begin{cor}\label{cor:def+s*}
If $A$ is a definable set in $Hom(W(X),H)$ then $s_*A$ is also a
definable set.
\end{cor}


The similar relation takes place between definable sets, $\cal
H$-closed filters and  maps $\widetilde s$ (\cite{BI-90-I},
\cite{Plotkin_AG}, \cite{Seven}).

\begin{prop}\label{prop:Gal+tilde-s}
 Let $T$ be a set of formulas from $\Phi(Y)$, $A$ be a set of points in
$Hom(W(Y),H)$. Let a homomorphism of free algebras $s: W(X)\to
W(Y)$ be given. Then
\begin{enumerate}
\item %
$\widetilde s ( A^{L}_{\cal H}) =\big( \widetilde s A
\big)^{L}_{\cal H}$,

\item %
$\widetilde s (T^{L}_{\cal H})\subseteq \big(\widetilde s T
\big)^{L}_{\cal H}$.

\end{enumerate}
\end{prop}

\section{Knowledge base model}\label{sec:KB-model}

In this section we will introduce the concept of a \emph{knowledge
base model} under consideration. But let us start with discussion
what is \emph{knowledge}.

\subsection{What is knowledge?}\label{sec:What-is-knowledge}


Although \emph{knowledge} is one of the most familiar concept, the
fundamental question about it: \emph{``What is it?"}.
Rarely this question has  been answered directly. Numerous papers
introduce one or another definition of knowledge, depending on
needs of a particular research and field of interest (see
\cite{Alavi-Leidner-Rew}, \cite{DavenportPrusak},
\cite{DavisShrobeSzolovits93},
\cite{FrawleyPiatetsky-ShapiroMatheus}, \cite{Helbig},
\cite{McInerney}, etc.)

Speaking about  knowledge we proceed from its representation in
three components.
\begin{enumerate}
\item %
\emph{Subject area of  knowledge},

\item %
\emph{Description of knowledge},

\item %
\emph{Content of  knowledge}.
\end{enumerate}

Let us  describe these component in more details.

\emph{Subject area of  knowledge} is presented by a model ${\cal
H}=(H,\Psi,f)$, where
\begin{itemize}
\item %
 $H$ is an algebra in fixed variety of algebras $\Theta$.

\item %
 $\Psi$ is a set of relation symbols $\varphi$.

\item %
 $f$ is an interpretation of each symbol $\varphi$ in $H$.
\end{itemize}

\emph{Description of knowledge} presents a syntactical component
of  knowledge. From algebraic viewpoint description of knowledge
is a set of formulas $T$, more precise, it is an $\cal H$-filter
in the algebra of formulas $\Phi(X)$, $X=\{x_1, \ldots , x_n\}$.

\emph{Content of  knowledge} is a subset in $H^{n}$, where $H^{n}$
is the Cartesian power of $H$. Each content of a knowledge $A$
corresponds to the description of a knowledge $T\subset\Phi(X)$,
$|X|=n$. If we regard $H^{n}$ as an affine space then  this
correspondence can be treated geometrically via Galois
correspondence. 

In order to describe the dynamic nature of a knowledge base two
categories and a functor are introduced: \emph{the category  of
knowledge description $F_{\Theta}(\cal H)$}, \emph{the category of
knowledge content $D_\Theta(\cal H)$} and \emph{the knowledge
functor $Ct_{\cal H}$}.

An object $F_\Theta^{X}({\cal H})$ of the category of knowledge
description $F_\Theta({\cal H})$ is the lattice of all $\cal
H$-closed filters in the algebra $\Phi(X)$, $X\in \Gamma$.

\begin{remark}
We cannot say that the usual set-theoretical union of ${\cal
H}$-closed filters is an ${\cal H}$-closed filter. To constitute a
lattice of ${\cal H}$-closed filters in $\Phi(X)$ there was
introduced a new operation
$$
T_1\overline \cup T_2=(T_1\cup T_2)^{LL}_{\cal H}.
$$
Then all ${\cal H}$-closed filters in $\Phi(X)$ form a lattice
with the operation $\overline \cup$ and $\cap$ (for details see
\cite{BI-90-I}, \cite{Plotkin_AG}, \cite{Plotkin_IsotAlg}).
\end{remark}

\subsection{Category of knowledge description $F_\Theta({\cal H})$}\label{sec:CatKnDesc}

An object $F_\Theta^{X}({\cal H})$ of the category $F_\Theta({\cal
H})$ is the lattice of all $\cal H$-closed filters in the algebra
$\Phi(X)$, $X\in \Gamma$.

Let a homomorphism $s:W(X)\to W(Y)$ and $\cal H$-closed filters
$T_1\in \Phi(X)$ and $T_{2}\in \Phi(Y)$ be given. We say that a
map $[s_{*}]: T_{1}\to T_{2}$ is \emph{admissible},  if
$s_{*}T_1\subseteq T_2$. Remind that $s_{*}$ is a map between
$\cal H$-closed filters in $\Phi(X)$ and $\Phi(Y)$ induced by the
corresponding morphism of the category $\widetilde \Phi$  (see
Section~\ref{sec:HalmosCat}).

A morphism between objects $F_\Theta^X({\cal H})$ and
$F_\Theta^Y({\cal H}) $
 $$
[s_{*}]: F_\Theta^X({\cal H}) \to F_\Theta^Y({\cal H})
 $$
is defined, if $[s_{*}]: T_1 \to T_2$ is admissible for every
$T_1\in F_\Theta^X({\cal H})$. 

We define a composition of morphisms $[s^{1}_{*}]:
F_\Theta^X({\cal H}) \to F_\Theta^Y({\cal H})$ and $[s^{2}_{*}]:
F_\Theta^Y({\cal H}) \to F_\Theta^Z({\cal H})$ as follows
$$
[s^{2}_{*}]\circ [s^{1}_{*}]=[s^{2}_{*}s^{1}_{*}].
$$
This definition is correct. Indeed, if $[s^{1}_{*}]: T_1 \to T_2$
and $[s^{2}_{*}]: T_2 \to T_3$, then $s^{1}_{*}T_1\subseteq T_2$
and $s^{2}_{*}T_2\subseteq T_3$. This means that
$s^{2}_{*}s^{1}_{*}T_1\subseteq T_3$ and $[s^{2}_{*}s^{1}_{*}]$ is
admissible for $T_1$ and $T_3$.

\subsection{Category of knowledge content $D_\Theta({\cal H})$}\label{sec:CatKnCont}

An object $D_\Theta^{X}({\cal H})$ of the category $D_\Theta({\cal
H})$ is the lattice of all definable sets in the affine space
$Hom(W(X),H)$, $X\in \Gamma$.

Let a homomorphism $s:W(X)\to W(Y)$ and definable sets $A_2\in
Hom(W(Y), H)$ and $A_{1}\in Hom(W(X), H)$ be given. We say that a
map $[\widetilde s]: A_{2}\to A_{1}$ is \emph{admissible},  if
$\widetilde s A_2\subseteq A_1$. Remind that $\widetilde s$  is a
map between definable sets in $Hom(W(Y),H)$ and $Hom(W(X),H)$
induced by the corresponding morphism of the category of affine
spaces $\Theta^{0}(H)$ (see Section~\ref{sec:BasicNotions}).

A morphism between objects $D_\Theta^Y({\cal H})$ and
$D_\Theta^X({\cal H}) $
 $$
[\widetilde s]: D_\Theta^Y({\cal H}) \to D_\Theta^X({\cal H})
 $$
is defined, if $[\widetilde s]: A_2 \to A_1$ is admissible for
every $A_2\in D_\Theta^Y({\cal H})$.

We define a composition of morphisms $[\widetilde{s^{1}}]:
F_\Theta^Z({\cal H}) \to D_\Theta^Y({\cal H})$ and
$[\widetilde{s^{2}}]: D_\Theta^Y({\cal H}) \to D_\Theta^X({\cal
H})$ as follows
$$
[\widetilde{s^{2}}]\circ [\widetilde{s^{1}}]=[\widetilde{s^{2}}
\widetilde{s^{1}}].
$$
This definition is correct. Indeed, if $[\widetilde{s^{1}}]: A_3
\to A_2$ and $[\widetilde{s^{2}}]: A_2 \to A_1$, then
$\widetilde{s^{1}}A_3\subseteq A_2$ and $\widetilde{s^{2}}
A_2\subseteq A_1$. This means that $\widetilde{s^{2}}
\widetilde{s^{1}}A_3\subseteq A_1$ and $[\widetilde{s^{2}}
\widetilde{s^{1}}]$ is admissible for $A_3$ and $A_1$.

\subsection{The knowledge functor $Ct_{\cal
H}$}\label{sec:KnFunctor}

The category of knowledge description $F_{\Theta}(\cal H)$ and the
category of knowledge content $D_{\Theta}(\cal H)$ are related by
the knowledge functor (for details see \cite{BI-90-I}).
$$
Ct_{\cal H}: F_{\Theta}({\cal H})\to D_{\Theta}({\cal H}),
$$
which is defined on objects by
$$
Ct_{\cal H}(F^{X}_{\Theta}({\cal H}))=D^{X}_{\Theta}({\cal H}),
$$
and on  morphisms by
$$
Ct_{\cal H}([s_{*}])=[\widetilde s],
$$
where $s:W(X)\to W(Y)$ is a given homomorphism of free algebras.
Moreover, if
 $[s_{*}]: F^{X}_{\Theta}({\cal H})\to F^{Y}_{\Theta}({\cal
H}) $ is a morphism in $F_{\Theta}({\cal H})$, such that
$$
[s_{*}]:T_1\to T_2,
$$
then $[\widetilde s]: D^{Y}_{\Theta}({\cal H}) \to
D^{X}_{\Theta}({\cal H})$ is  a morphism in $D_{\Theta}({\cal H})$
defined by the rule
$$
[\widetilde{s}]: (T_2)^{L}_{\cal H} \to (T_{1})^{L}_{\cal H}.
$$

Now we are at the point to give a definition of  knowledge base
model. Let a model ${\cal H}=(H,\Psi,f)$ be given.

{\sloppy
\begin{defin}\label{def:KB}
A knowledge base $KB = KB(H,\Psi,f)$ is a triple $\big(
F_\Theta({\cal H}), D_\Theta({\cal H}),Ct_{\cal H} \big)$, where
$F_\Theta({\cal H})$ is the category of knowledge description,
$D_\Theta({\cal H})$ is the category of knowledge content, and
$$
Ct_{\cal H}: F_\Theta({\cal H}) \to D_\Theta({\cal H})
$$
is the contravariant functor.
\end{defin}

}

\begin{remark}
We will use the term ``a knowledge base'' instead of a more
precise ``a knowledge base model''.
\end{remark}


One can say that defined knowledge base model is a sort of
automaton (see \cite{P-DB}), where queries are objects of the
category of knowledge descriptions  $F_\Theta(\cal H)$, replies
are objects of the category of knowledge content $D_\Theta(\cal
H)$. To be such automaton a knowledge base also presupposes a
connection with a particular data (information). This information
is held in the subject area presented by the model ${\cal
H}=(H,\Psi,f)$.

The knowledge  functor $Ct_{\cal H}$ gives a dynamical passage
between  queries and replies, namely,  between categories
$F_\Theta(\cal H)$ and $D_\Theta(\cal H)$.
Moreover, 
this passage is one-to-one correspondence.

\begin{theorem}[\cite{BI-90-II}]\label{th:dual:CatF-CatD}
The knowledge functor $Ct_{\cal H}$ gives rise to the dual
isomorphism between the category of knowledge description
$F_{\Theta}({\cal H})$ and the category of knowledge content
$D_\Theta({\cal H})$. 
\end{theorem}

\section{Knowledge bases equivalences}\label{sec:KB-equiv}


\subsection{An overview}\label{sec:KB-overview}

In this section we give a short review of our previous results
about various knowledge bases equivalences and connections between
them. We start with the most strong equivalence from algebraic
viewpoint, namely, with \emph{isomorphic knowledge bases}.

Fix a variety of algebras $\Theta$, algebras $H_1$ and $H_2$ from
$\Theta$ and a set of relation symbols $\Psi$.  Let two models
${\cal H}_1=(H_1,\Psi, f_1)$ and ${\cal H}_2=(H_2,\Psi, f_2)$   be
given.

\begin{defin}\label{def:Isom-KB}
Two knowledge bases $KB({\cal H}_1)$ and $KB({\cal H}_2)$ are
called  isomorphic if the corresponding models ${\cal H}_1$ and
${\cal H}_2$ are isomorphic.
\end{defin}

The notion of isomorphic knowledge bases is very strong. It
presuppose an isomorphism of subject areas of knowledge bases,
which automatically implies an isomorphism of categories of
knowledge description of corresponding knowledge bases and,
according to Theorem~\ref{th:dual:CatF-CatD}, an isomorphism of
categories of knowledge content of corresponding knowledge bases.

For practical needs there is more appropriate and not too strong
notion of \emph{informationally equivalent knowledge bases}.


Let two models ${\cal H}_1=(H_1,\Psi,f_1)$ and ${\cal
H}_2=(H_2,\Psi,f_2)$ be given. Take  the corresponding knowledge
bases $KB({\cal H}_1)$ and $KB({\cal H}_2)$.

\begin{defin}\label{def:InformEquivKB}
Knowledge base $KB({\cal H}_1)$ and $KB({\cal H}_2)$ are called
informationally equivalent, if the categories of knowledge
description $F_\Theta({\cal H}_1)$ and $F_\Theta({\cal H}_2)$ are
isomorphic.
\end{defin}

\begin{remark}\label{rem:Def-InformEquiv-KB}
In view of Theorem~\ref{th:dual:CatF-CatD}, the categories of
knowledge description $F_\Theta({\cal H}_1)$ and $F_\Theta({\cal
H}_2)$ are isomorphic if and only if the categories of knowledge
content $D_\Theta({\cal H}_1)$ and $D_\Theta({\cal H}_2)$ are
isomorphism. Thus, one can formulate
Definition~\ref{def:InformEquivKB} in terms of isomorphism of the
categories of knowledge content.
\end{remark}

In plain words, the informational equivalence of knowledge bases
means that everything that can be asked from one knowledge base
can be asked from the other and to conclude that this is the same
information.

Our main interest is to find
out how the informational equivalence is related to the logical
description of knowledge bases. In this concern, there were
defined  \emph{elementarily equivalent, logically-geometrical
equivalent, $LG$-isotypic  knowledge bases and others}.

The notion of $LG$-equivalence (logically-geometrical equivalence)
of knowledge bases is based on geometrical approach, whereas
$LG$-isotypic knowledge bases are defined using logical tools. But
these notions
give us the same description of knowledge bases: 

\begin{theorem}[\cite{APP}]\label{th:LG-equiv-then-InformEquiv}
Logically-geometrical equivalent (or $LG$-isotypic) knowledge
bases are informationally equivalent.
\end{theorem}

In the next section 
we will deal with one more equivalence for knowledge bases, which
is defined using category theory tools.

\subsection{Logically automorphically equivalent knowledge
bases}\label{sec:LogAutEquiv-KB}

As we have seen the notion of $LG$-equivalent and $LG$-isotypic
knowledge bases is good enough to distinguish two knowledge bases.
That is, $LG$-equivalent and $LG$-isotypic
knowledge bases are informationally equivalent. 

In this section we introduce 
\emph{logically automorphical equivalence of knowledge bases}. We
will see that this notion also gives a good characterization of
knowledge bases.

Let us start with some preliminary constructions.

\subsubsection{Functor $Cl_{\cal H}$}\label{sec:Funct-Cl}



The functor $Cl_{\cal H}$ presents a connection between the
category $\widetilde\Phi$ and the category $F_{\Theta}$ of
lattices of all closed filters, for various models ${\cal
H}_{i}=(H_{i},\Psi,f_{i})$ with algebras $H_{i}$ from a variety
$\Theta$ defined as follows.


An object $F_{\Theta}^{X}({\cal H}_i)$ of the category
$F_{\Theta}$ is the lattice of all ${\cal H}_{i}$-closed filters
in $\Phi(X)$, $X\in \Gamma$.

Morphisms of the category $F_{\Theta}$ are maps of lattice of all
closed filters, which preserve partial order on corresponding
objects, but they not to be necessarily homomorphisms of lattices.

\begin{remark}
In Section~\ref{sec:KB-model} we have  defined the category
$F_\Theta({\cal H}_{i})$ of lattices of all ${\cal H}_{i}$-closed
filters (the category of knowledge description), that is, this is
the category over a fixed model ${\cal H}_{i}$.

Thus, $F_\Theta({\cal H}_{i})$ is a full subcategory of
$F_{\Theta}$ and, hence, morphisms of the category
$F_{\Theta}({\cal H}_i)$ are morphisms of $F_{\Theta}$. But there
are other morphisms in $F_{\Theta}$, we will do not specify them.
For example, there are morphism between objects $F_\Theta^X({\cal
H}_1)$ and $F_\Theta^X({\cal H}_2)$, where ${\cal
H}_1=(H_1,\Psi,f_1)$, ${\cal H}_2=(H_2,\Psi,f_2)$ are models with
different algebras $H_1$ and $H_2$ from $\Theta$.
\end{remark}

One can define a correspondence 
$$
Cl_{\cal H}:\widetilde \Phi\to F_{\Theta},
$$
on objects as follows:
$$
Cl_{\cal H}(\Phi(X))=F_{\Theta}^{X}({\cal H}),
$$
and if $s_{*}:\Phi(X)\to \Phi(Y)$ is a morphism in $\widetilde
\Phi$, then
$$
Cl_{\cal H}(s_{*})=[s_*]^{0}: F_{\Theta}^{X}({\cal H})\to
F_{\Theta}^{Y}({\cal H})
$$
is a morphism in $F_{\Theta}$, such that
$$
[s_{*}]^{0}: (T)^{LL}_{{\cal H}} \to \big(s_{*}T\big)^{LL}_{{\cal
H}},
$$
where $T\subset \Phi(X)$, $T^{LL}_{\cal H}$ is an $\cal H$-closed
filter in $F_{\Theta}^{X}({\cal H})$,  $\big(s_{*}T
\big)^{LL}_{\cal H}$ an $\cal H$-closed filter in
$F_{\Theta}^{Y}({\cal H})$.

Next diagram 
illustrates the correspondence $Cl_{H}$:
$$
 \xymatrix{
 \Phi(X)  \ar[rr] \ar[dd]_{s_{*}}& &  F^{X}_{\Theta}({\cal H}) \ar[dd]^{[s_{*}]^{0}}
 \ar@/^3pc/[dd]^{\dots}  \ar@/^5pc/[dd]^{[s_{*}]^{i} \dots}\\
 \ar@{-}[r] & Cl_{H} \ar[r] & \\
 \Phi(Y)  \ar[rr] & &  F^{Y}_{\Theta}({\cal H}), \\
 }
$$
where  $[s_{*}]^{i}$ are some other morphisms between objects
$F^{X}_{\Theta}({\cal H})$ and $F^{Y}_{\Theta}({\cal H})$
associated with morphism $s_{*}$ of the category $\widetilde
\Phi$.

The following proposition takes place.

\begin{prop}\label{prop:Functor-Cl}
The correspondence $Cl_{\cal H}: \widetilde \Phi\to F_{\Theta}$ is
a covariant functor.
\end{prop}

\begin{proof}
If $s_{*}=id_{\Phi(X)}:\Phi(X)\to \Phi(X)$ is the identity
morphism of the object $\Phi(X)$, then $s_{*}T=T$ for every
$T\subset \Phi(X)$. Thus,  $Cl_{\cal
H}(id_{\Phi(X)})=[id_{\Phi(X)}]^{0}=id_{F^{X}_{\Theta}({\cal H})}$
is the identity morphism of the object $F^{X}_{\Theta}({\cal H})$.

 Let $s^{1}_{*}:\Phi(X)\to \Phi(Y)$, $s^{2}_{*}:\Phi(Y)\to
\Phi(Z)$ be morphisms in $\widetilde \Phi$. Take  a subset $T$
from $\Phi(X)$. Then
$$
s_{*}^{1}: T\to s_{*}^{1}T, 
$$
$$
s_{*}^{2}: s_{*}^{1}T \to s_{*}^2(s_{*}^{1}T)
$$
and
\begin{equation}\label{eq:Cl(s1s2)}
 Cl_{\cal H}(s^{2}_{*}\circ s^{1}_{*}): T^{LL}_{\cal H} \to
(s^{2}_{*}s^{1}_{*}T)^{LL}_{\cal H}.
\end{equation}
From the other hand,
$$
Cl_{\cal H}(s^1_{*})=[s_{*}^{1}]^{0}: T^{LL}_{\cal H}\to \big(
s_{*}^{1}T \big)^{LL}_{\cal H}$$
 and
$$
Cl_{\cal H}(s^2_{*})=[s_{*}^{2}]^{0}:
\big(s_{*}^{1}T\big)^{LL}_{\cal H} \to
\big(s_{*}^2(s_{*}^{1}T)^{LL}_{\cal H}\big)^{LL}_{\cal H}.
$$
Therefore,
$$
Cl_{\cal H}(s^1_{*})\circ Cl_{\cal H}(s^2_{*})=
[s_{*}^{1}]^{0}\circ [s_{*}^{2}]^{0}: T^{LL}_{\cal H}\to
\big(s_{*}^2(s_{*}^{1}T)^{LL}_{\cal H}\big)^{LL}_{\cal H}.
$$
Let us simplify the right part of the last equation. By
Proposition~\ref{prop:Gal+s*}, the equality $(s_{*}T)^{L}_{\cal
H}=s_{*} T^{L}_{\cal H}$ takes place. Thus,
$$
\big(s_{*}^2(s_{*}^{1}T)^{LL}_{\cal H}\big)^{LL}_{\cal
H}=\big(s_{*}^2\big((s_{*}^{1}T)^{LL}_{\cal H}\big)^{L}_{\cal
H}\big)^{L}_{\cal H}.
$$
%
Using the property of the Galois correspondence, namely,
$T^{LLL}_{\cal H}=T^{L}_{\cal H}$, we have
$$
\big(s_{*}^2\big((s_{*}^{1}T)^{LL}_{\cal H}\big)^{L}_{\cal
H}\big)^{L}_{\cal H}=\big( s^{2}_{*}(s^{1}_{*}T)^{L}_{\cal
H}\big)^{L}_{\cal H}.
$$
Applying again the equality $(s_{*}T)^{L}_{\cal H}=s_{*}
T^{L}_{\cal H}$, we
get 
$$
\big( s^{2}_{*}(s^{1}_{*}T)^{L}_{\cal H}\big)^{L}_{\cal
H}=(s^{2}_{*}s^{1}_{*}T)^{LL}_{\cal H}.
$$
Thus,
 \begin{equation}\label{eq:Cl(s1)Cl(s2)}
  Cl_{\cal H}(s^1_{*})\circ Cl_{\cal
H}(s^2_{*})= [s_{*}^{1}]^{0}\circ [s_{*}^{2}]^{0}: T^{LL}_{\cal
H}\to (s^{2}_{*}s^{1}_{*}T)^{LL}_{\cal H}.
\end{equation}
Comparing  equations (\ref{eq:Cl(s1s2)}) and
(\ref{eq:Cl(s1)Cl(s2)}), we conclude that
$$
Cl_{\cal H}(s^{2}_{*}\circ s^{1}_{*})=Cl_{\cal H}(s^{2}_{*}) \circ
Cl_{\cal H}(s^{1}_{*}),
$$
and $Cl_{\cal H}$ is a covariant functor.
\end{proof}

\subsubsection{Definition of logically automorphically equivalence}

We will use the notion of \emph{isomorphism}  of two functors
(\emph{natural isomorphism} in terms of \cite{MacLane}).

\begin{defin}\label{def:isom-of-functors}
Let ${\cal F}_1$ and ${\cal F}_2$ be functors from a category
${\cal C}_1$ to a category ${\cal C}_2$. An isomorphism $\alpha:
{\cal F}_1 \to {\cal F}_2 $ of functors ${\cal F}_1$ and ${\cal
F}_2$ is a function which assigns to each object $C$ in ${\cal
C}_1$ a two-sided morphism $\alpha(C): {\cal
F}_1(C)\leftrightarrow {\cal F}_2(C)$ in the category ${\cal C}_2$
in such a way that for every morphism $\nu:C\to C'$ of the
category ${\cal C}_1$ the diagram is commutative:
$$
\xymatrix{
{\cal F}_1(C) \ar@{<->}[r]^{\alpha(C)} \ar[d]_{{\cal F}_1(\nu)} & {\cal F}_2(C) \ar[d]^{{\cal F}_2(\nu)}\\
{\cal F}_1(C') \ar@{<->}[r]^{\alpha(C')} & {\cal F}_2(C'). }
$$
\end{defin}

 In Section~\ref{sec:Funct-Cl} we construct the covariant
functor
$$
Cl_{\cal H}: \widetilde \Phi \to F_{\Theta},
$$
 where ${\cal
H}=(H,\Psi, f)$ is a model, $\widetilde \Phi$ is the category of
algebras of formulas, $F_{\Theta}$ is the category of lattices of
closed filters. Using this functor we define a notion of
\emph{logically automorphical equivalence for models}.

Let two models ${\cal H}_1=(H_1,\Psi, f_1)$ and ${\cal
H}_2=(H_2,\Psi, f_2)$ be given and let $\varphi$ be an
automorphism of the category $\widetilde \Phi$.

\sloppy{
\begin{defin}\label{def:AutomEquiv-of-models}
Models  ${\cal H}_1=(H_1, \Psi, f_1)$ and  ${\cal H}_2=(H_2, \Psi,
f_2)$ 
are called logically automorphically equivalent if for some
automorphism $\varphi$ of the category $\widetilde \Phi$ there is
the functor isomorphism
$$
\alpha_\varphi :Cl_{{\cal H}_1} \to Cl_{{\cal H}_2} \cdot \varphi.
$$
\end{defin}
} 

This definition gives rise to the notion of logically
automorphically equivalent knowledge bases. Let two models ${\cal
H}_1=(H_1,\Psi,f_1)$, ${\cal H}_2=(H_2,\Psi,f_2)$ and the
corresponding knowledge bases $KB({\cal H}_1)$ and $KB({\cal
H}_2)$  be given.

\begin{defin}\label{def:AutomEquiv-KB}%
Knowledge bases $KB({\cal H}_1)$ and $KB({\cal H}_2)$ are called
logically automorphically equivalent if the corresponding models
${\cal H}_1$ and ${\cal H}_2$ are logically automorphically
equivalent.
\end{defin}

\subsubsection{Auxiliary constructions}\label{sec:PrelConst-LogAutEquiv-KB}

In this section we present 
results which we will use to prove the main result about logically
automorphically equivalent knowledge bases.

Let two logically automorphically equivalent models  ${\cal
H}_1=(H_1, \Psi, f_1)$ and  ${\cal H}_2=(H_2, \Psi, f_2)$ be
given.

Logically automorphical equivalence of the models ${\cal
H}_1=(H_1, \Psi, f_1)$ and  ${\cal H}_2=(H_2, \Psi, f_2)$ means
that there exists an automorphism $\varphi$ of the category
$\widetilde \Phi$, such that the functors $Cl_{{\cal H}_1}$ and
$Cl_{{\cal H}_2}\cdot \varphi$ are isomorphic. This fact implies
that there is the commutative diagram (see
Definition~\ref{def:isom-of-functors}):
\begin{equation}\label{diag:aut-equiv-1}
\CD
Cl_{{\cal H}_1} (\Phi(X)) @> \alpha_{\varphi}(\Phi(X)) >> Cl_{{\cal H}_2}\cdot \varphi (\Phi(X))\\
@V Cl_{{\cal H}_1}(s_{*})  VV @VV Cl_{{\cal H}_2}\cdot \varphi (s_{*}) V\\
Cl_{{\cal H}_1} (\Phi(Y)) @> \alpha_{\varphi}(\Phi(Y)) >>
Cl_{{\cal H}_2}\cdot \varphi (\Phi(Y)),
\endCD
\end{equation}
where $\alpha_{\varphi}:Cl_{{\cal H}_1}\to Cl_{{\cal H}_2}\cdot
\varphi$ is an isomorphism of functors, $\Phi(X)$ and  $\Phi(Y)$
are objects of the category $\widetilde \Phi$ and $s_{*}:
\Phi(X)\to \Phi(Y)$ is a morphism in $\widetilde \Phi$.

 Recall that
$$
Cl_{{\cal H}_i} (\Phi(X))=F_{\Theta}^{X}({\cal H}_i),
$$
and if $s_{*}:\Phi(X)\to \Phi(Y)$, then
$$
Cl_{{\cal H}_i}(s_{*})=[s_{*}]^{0}_{{\cal H}_i}:
F_{\Theta}^{X}({\cal H}_i)\to F_{\Theta}^{Y}({\cal H}_i),
$$
such that
$$
[s_{*}]^{0}_{{\cal H}_i}: T^{LL}_{{\cal H}_i}\to \big(s_{*}T
\big)^{LL}_{{\cal H}_i},
$$
where $T\subset\Phi(X)$, $T^{LL}_{{\cal H}_i}$ is an ${{\cal
H}_i}$-closed filter in $F_{\Theta}^{X}({{\cal H}_i})$,
$\big(s_{*}T \big)^{LL}_{{\cal H}_i}$ an ${{\cal H}_i}$-closed
filter in $F_{\Theta}^{Y}({{\cal H}_i})$, for more details see
Section~\ref{sec:Funct-Cl}.

\begin{remark}
We add subscribe index ${\cal H}_i$ for morphism $[s_{*}]^{0}$ in
order to distinguish morphisms in $F_{\Theta}({\cal H}_1)$ and
$F_{\Theta}({\cal H}_2)$.
\end{remark}

Let $\varphi$ be an automorphism of $\widetilde \Phi$, such that
$$
\varphi(\Phi(X))=\Phi(X') \mbox{ and } \
\varphi(\Phi(Y))=\Phi(Y'),
$$
where $X,Y,X',Y'\in \Gamma$. In particular,  this means that if
$s_{*}:\Phi(X)\to \Phi(Y)$, then $\varphi(s_{*}):\Phi(X')\to
\Phi(Y')$.

Using the settings above,  we have
$$
Cl_{{\cal H}_2}\cdot \varphi (\Phi(X))= Cl_{{\cal H}_2}
(\Phi(X'))=F_{\Theta}^{X'}({\cal H}_2),
$$
$$
Cl_{{\cal H}_2}\cdot \varphi (\Phi(Y))= Cl_{{\cal H}_2}
(\Phi(Y'))=F_{\Theta}^{Y'}({\cal H}_2),
$$
$$
Cl_{{\cal H}_2}\cdot \varphi (s_{*})= Cl_{{\cal H}_2}(\varphi
(s_{*}))=[\varphi (s_{*})]^{0}_{{\cal H}_2}.
$$

We can rewrite diagram~(\ref{diag:aut-equiv-1}) as follows
\begin{equation}\label{diag:aut-equiv-2}
\CD
F_{\Theta}^{X}({\cal H}_1) @> \alpha_{\varphi} >> F_{\Theta}^{X'}({\cal H}_2)\\
@V [s_{*}]^{0}_{{\cal H}_1}  VV @VV [\varphi (s_{*})]^{0}_{{\cal H}_2} V\\
F_{\Theta}^{Y}({\cal H}_1) @> \alpha_{\varphi} >>
F_{\Theta}^{Y'}({\cal H}_2).
\endCD
\end{equation}

\begin{remark}
Here and later on we will write simply $\alpha_{\varphi}$ instead
of $\alpha_{\varphi}(\Phi(X))$ or $\alpha_{\varphi}(\Phi(Y))$.
\end{remark}

The next proposition provides a connection between morphism in
categories $F_{\Theta}({\cal H}_1)$ and $F_{\Theta}({\cal H}_2)$
over logically automorphically equivalent models  ${\cal H}_1$ and
${\cal H}_2$.


\begin{prop}\label{prop:AutEquivModel-AdmissibleSets}
Let two logically automorphically equivalent models  ${\cal
H}_1=(H_1, \Psi, f_1)$ and  ${\cal H}_2=(H_2, \Psi, f_2)$ be
given. The map $[s_{*}]_{{\cal H}_1}: T_1\to T_2$ is admissible if
and only if the map $[\varphi(s_{*})]_{{\cal H}_2}:
\alpha_{\varphi}T_1 \to \alpha_{\varphi}T_2 $ is admissible, where
$T_1\in F^{X}_{\Theta}({\cal H}_1)$, $T_2\in F^{Y}_{\Theta}({\cal
H}_1)$.
\end{prop}

\begin{proof}
Diagram~(\ref{diag:aut-equiv-2}) gives rises  to the following
diagram:
$$
 \xymatrix{
T_1  \ar@{<->}[r]^{\alpha_{\varphi}} \ar[d]_{[s_{*}]^{0}_{{\cal
H}_1}} & \alpha_{\varphi} T_1 \ar[d]^{[\varphi(s_{*})]^{0}_{{\cal
H}_2}}\\
(s_{*}T_1)^{LL}_{{\cal H}_1} \ar@{<->}[r]^(0.4){\alpha_{\varphi}}
&
\big(\varphi(s_{*})(\alpha_{\varphi}T_1)\big)^{LL}_{{\cal H}_2}. \\
 }
$$
By the definition, the map $[s_{*}]{{\cal H}_1}: T_1\to T_2$ is
admissible if and only if $s_{*}T_1\subseteq T_2$. Moreover,
$T_{2}$ is an ${\cal H}_1$-closed filter, hence
$\big(s_{*}T_1\big)^{LL}_{{\cal H}_1}\subseteq T_2$. Thus, we can
extend the diagram above as follows:
$$
 \xymatrix{
T_1  \ar@{<->}[r]^{\alpha_{\varphi}} \ar[d]_{[s_{*}]^{0}_{{\cal
H}_1}} \ar@/_5pc/[dd]_{[s_{*}]_{{\cal H}_1}} & \alpha_{\varphi}
T_1 \ar[d]^{[\varphi(s_{*})]^{0}_{{\cal
H}_2}}\\
(s_{*}T_1)^{LL}_{{\cal H}_1} \ar@{<->}[r]^(0.4){\alpha_{\varphi}}
\ar@{}[d]|{\mid\bigcap} &
\big(\varphi(s_{*})(\alpha_{\varphi}T_1)\big)^{LL}_{{\cal H}_2} \\
T_2  \ar@{<->}[r]^{\alpha_{\varphi}} & \alpha_{\varphi} T_2.
 }
$$
From the diagram follows that
$$
\big(\varphi(s_{*}) (\alpha_{\varphi}T_1)\big)^{LL}_{{\cal
H}_2}\subseteq \alpha_{\varphi} T_2.
$$
Moreover, $\varphi(s_{*}) (\alpha_{\varphi}T_1) \subseteq
\big(\varphi(s_{*}) (\alpha_{\varphi}T_1)\big)^{LL}_{{\cal H}_2}$
 and the map $[\varphi(s_{*})]_{{\cal H}_2}:
\alpha_{\varphi}T_1 \to \alpha_{\varphi}T_2 $ is admissible.

Thus, the map $[s_{*}]_{{\cal H}_1}: T_1\to T_2$ is admissible if
and only if the map $[\varphi(s_{*})]_{{\cal H}_2}:
\alpha_{\varphi}T_1 \to \alpha_{\varphi}T_2 $ is admissible. The
following diagram gives the illustration:
$$
 \xymatrix{
T_1  \ar@{<->}[r]^{\alpha_{\varphi}} \ar[d]_{[s_{*}]^{0}_{{\cal
H}_1}} \ar@/_5pc/[dd]_{[s_{*}]_{{\cal H}_1}} & \alpha_{\varphi}
T_1 \ar[d]^{[\varphi(s_{*})]^{0}_{{\cal
H}_2}} \ar@/^6pc/[dd]^{[\varphi (s_{*})]_{{\cal H}_2}}\\
(s_{*}T_1)^{LL}_{{\cal H}_1} \ar@{<->}[r]^(0.4){\alpha_{\varphi}}
\ar@{}[d]|{\mid\bigcap} &
\big(\varphi(s_{*})(\alpha_{\varphi}T_1)\big)^{LL}_{{\cal H}_2} \ar@{}[d]|{\mid\bigcap} \\
T_2  \ar@{<->}[r]^{\alpha_{\varphi}} & \alpha_{\varphi} T_2.
 }
$$

\end{proof}

\subsubsection{The main result}\label{sec:MainTh-LogAutEquiv-KB}

In this section we present the main result of this paper, namely,
we will show that logically automorphically equivalent knowledge
bases  are informationally equivalent.

The following preliminary result describes relation between
logically automorphically equivalent models and categories of
lattices of $\cal H$-closed filters over corresponding models.

In this section we will use notations from the previous section.

\begin{theorem}\label{th:Aut-Equiv-Model}
If models  ${\cal H}_1=(H_1, \Psi, f_1)$ and  ${\cal H}_2=(H_2,
\Psi, f_2)$ are logically automorphically equivalent, then the
categories $F_\Theta({\cal H}_1)$ and $F_\Theta({\cal H}_2)$ are
isomorphic.
\end{theorem}

\begin{proof}
To prove the theorem we will construct a correspondence (functor)
$$
{\cal F}: F_\Theta({\cal H}_1) \to F_\Theta({\cal H}_2)
$$
and show that it gives rise to an isomorphism of the given 
categories.


For an object $F_{\Theta}^{X}({\cal H}_1)$ from $F_\Theta({\cal
H}_1)$ we set
$$
{\cal F}(F_{\Theta}^{X}({\cal H}_1))=F_{\Theta}^{X'}({\cal H}_2),
$$
where $X$ and $X'$ are correlated by the given automorphism
$\varphi$, that is, $\varphi(\Phi(X))=\Phi(X')$ (see
Section~\ref{sec:PrelConst-LogAutEquiv-KB}).

Let $[s_{*}]_{{\cal H}_1}: F_{\Theta}^{X}({\cal H}_1)\to
F_{\Theta}^{Y}({\cal H}_1)$ be a morphism in $F_{\Theta}({\cal
H}_1)$, such that
$$
[s_{*}]_{{\cal H}_1}: T_1\to T_2,
$$
where $T_{1}\in F_{\Theta}^{X}({\cal H}_1)$, $T_{2}\in
F_{\Theta}^{Y}({\cal H}_1)$ and  $s_{*}T_1\subseteq T_2$.

We determine the morphism 
$$
{\cal F}([s_{*}]_{{\cal H}_1}): F_{\Theta}^{X'}({\cal H}_2)\to
F_{\Theta}^{Y'}({\cal H}_2)
$$
by the rule 
$$
{\cal F}([s_{*}]_{{\cal H}_1})=[\varphi(s_{*})]_{{\cal H}_2}:
\alpha_{\varphi}T_1 \to \alpha_{\varphi}T_2.
$$
By assumption, $\alpha_{\varphi}$ is a bijection (two-sided
morphism) between objects $F_{\Theta}^{X}({\cal H}_1)$ and
$F_{\Theta}^{X'}({\cal H}_2)$, so if $T$ runs all ${\cal
H}_1$-closed filters from $F_{\Theta}^{X}({\cal H}_1)$, then
$\alpha_{\varphi}(T)$ runs all ${\cal H}_2$-closed filters from
$F_{\Theta}^{X'}({\cal H}_2)$. Moreover, in view of
Proposition~\ref{prop:AutEquivModel-AdmissibleSets},  the morphism
$[\varphi(s_{*})]_{{\cal H}_2}$ is defined correctly, that is,
$\varphi(s_{*})(\alpha_{\varphi}T_1) \subseteq
\alpha_{\varphi}T_2$.

Let us show that defined  in such a way correspondence ${\cal F}$
is, indeed, a functor.

Denote by $id_{F^{X}_{\Theta}({\cal H}_1)}$ the identity morphism
of the object $F^{X}_{\Theta}({\cal H}_1)$. By definition of $\cal
F$ we have the diagram:
$$
 \xymatrix{
T  \ar@{<->}[r]^{\alpha_{\varphi}}
\ar[d]_{id_{F^{X}_{\Theta}({\cal H}_1)}} & \alpha_{\varphi} T \ar[d]^{{\cal F}( id_{F^{X}_{\Theta}({\cal H}_1)})}\\
T \ar@{<->}[r]^{\alpha_{\varphi}} &
\alpha_{\varphi} T, \\
 }
$$
where $T\in F^{X}_{\Theta}({\cal H}_1)$. Thus, ${\cal F}(
id_{F^{X}_{\Theta}({\cal H}_1)})$ is the identity morphism of
$F^{X'}_{\Theta}({\cal H}_2)={\cal F}( F^{X}_{\Theta}({\cal
H}_1))$.


Let two morphisms of $F_{\Theta}({\cal H}_1)$ be given:
$$
[s_{*}^{1}]_{{\cal H}_1}:F_{\Theta}^{X}({\cal H}_1)\to
F_{\Theta}^{Y}({\cal H}_1),
$$
$$
[s_{*}^{2}]_{{\cal H}_1}:F_{\Theta}^{Y}({\cal H}_1)\to
F_{\Theta}^{Z}({\cal H}_1).
$$
We will check that
$$
{\cal F}([s_{*}^{2}]_{{\cal H}_1}\circ [s_{*}^{1}]_{{\cal H}_1}])=
{\cal F}([s_{*}^{2}]_{{\cal H}_1}) \circ {\cal
F}([s_{*}^{1}]_{{\cal H}_1}]).
$$

Let $T_1$ be an ${\cal H}_1$-closed filter from
$F_{\Theta}^{X}({\cal H}_1)$ and
$$
[s_{*}^{1}]_{{\cal H}_1}: T_1\to T_2,
$$
$$
[s_{*}^{2}]_{{\cal H}_1}: T_2 \to T_3.
$$
Thus,
$$
[s_{*}^{2}]_{{\cal H}_1}\circ [s_{*}^{1}]_{{\cal H}_1} 
: T_1\to T_3
$$
and
\begin{equation}\label{eq:F(s2s1)}
{\cal F}([s_{*}^{2}]_{{\cal H}_1}\circ [s_{*}^{1}]_{{\cal H}_1}):
\alpha_{\varphi}(T_1) \to \alpha_{\varphi}(T_3).
\end{equation}
From the other hand,
$$
{\cal F }([s_{*}^{1}]_{{\cal H}_1}): \alpha_{\varphi} T_1 \to
\alpha_{\varphi} T_2
$$
and
$$
{\cal F}([s_{*}^{2}]_{{\cal H}_1}): \alpha_{\varphi} T_2 \to
\alpha_{\varphi} T_3.
$$
Consequently, the composition of functors ${\cal F
}([s_{*}^{1}]_{{\cal H}_1})$ and ${\cal F}([s_{*}^{2}]_{{\cal
H}_1})$ works as follows:
\begin{equation}\label{eq:F(s2)F(s1)}
{\cal F }([s_{*}^{2}]_{{\cal H}_1})\circ {\cal
F}([s_{*}^{1}]_{{\cal H}_1}): \alpha_{\varphi} T_1 \to
\alpha_{\varphi} T_3.
\end{equation}

 \noindent
 Summarizing equations~(\ref{eq:F(s2s1)}) and
(\ref{eq:F(s2)F(s1)}), we have
$$
{\cal F}([s_{*}^{2}]_{{\cal H}_1}\circ [s_{*}^{1}]_{{\cal
H}_1})={\cal F }([s_{*}^{2}]_{{\cal H}_1})\circ {\cal
F}([s_{*}^{1}]_{{\cal H}_1}).
$$
Thus, the correspondence ${\cal F}: F_{\Theta}({\cal H}_1) \to
F_{\Theta}({\cal H}_2)$ is a covariant functor.

Now we will construct the inverse functor
$$
{\cal F}': F_{\Theta}({\cal H}_2) \to F_{\Theta}({\cal H}_1).
$$
The functor ${\cal F}'$ works on objects as
$$
{\cal F}' \big( F^{X'}_{\Theta}({\cal H}_2) \big) =
F^{X}_{\Theta}({\cal H}_1),
$$
where $X'$ and $X$ are correlated by the given automorphism
$\varphi$, that is, $\varphi^{-1}(\Phi(X'))=\Phi(X)$ (see
Section~\ref{sec:PrelConst-LogAutEquiv-KB}).

Take  a morphism $s_{*}: \Phi(X')\to \Phi(Y')$ in the category
$\widetilde \Phi$. Let
$$
[s_{*}]_{{\cal H}_2}: F_{\Theta}^{X'}({\cal H}_2)\to
F_{\Theta}^{Y'}({\cal H}_2)
$$
 be a
morphism in  $F_{\Theta}({\cal H}_2)$, such that
$$
[s_{*}]_{{\cal H}_2}: T_1\to T_2,
$$
where $T_1\in F_{\Theta}^{X'}({\cal H}_2)$, $T_2 \in
F_{\Theta}^{Y'}({\cal H}_2)$ and $s_{*} T_1\subseteq T_2$.
 We determine the morphism
$$
{\cal F}'([s_{*}]_{{\cal H}_2}): F_{\Theta}^{X}({\cal H}_1)\to
F_{\Theta}^{Y}({\cal H}_1)
$$
 as
follows
\begin{equation}\label{eq:F'(s*)}
{\cal F}'([s_{*}]_{{\cal H}_2})=[\varphi^{-1}(s_{*})]_{{\cal
H}_1}:
 \alpha^{-1}_{\varphi}T_1\to \alpha^{-1}_{\varphi}T_2.
\end{equation}

 The map $\alpha^{-1}_{\varphi}$ is defined, since
$\alpha_{\varphi}$ is a  two-sided morphism 
 between objects $F_{\Theta}^{X}({\cal H}_1)$ and $F_{\Theta}^{X'}({\cal
H}_2)$. Moreover, if $T$ runs all ${\cal H}_2$-closed filters from
$F_{\Theta}^{X'}({\cal H}_2)$, then $\alpha^{-1}_{\varphi}(T)$
runs all ${\cal H}_1$-closed filters from $F_{\Theta}^{X}({\cal
H}_1)$. Hence, in view of
Proposition~\ref{prop:AutEquivModel-AdmissibleSets},  morphism
$[\varphi^{-1}(s_{*})]_{{\cal H}_1}$ is determined correctly and
the diagram takes place:
$$
 \xymatrix{
T_1  \ar@{<->}[r]^{\alpha_{\varphi}^{-1}} \ar[d]_{[s_{*}]_{{\cal
H}_2}}  & \alpha_{\varphi}^{-1} T_1
\ar[d]^{[\varphi^{-1}(s_{*})]_{{\cal
H}_1}} \\
T_2  \ar@{<->}[r]^{\alpha_{\varphi}^{-1}} & \alpha_{\varphi}^{-1}
T_2.
 }
$$

Let us show that defined  in such a way correspondence ${\cal F}
'$ is a functor.

If $id_{F^{X'}_{\Theta}({\cal H}_2)}$ the identity morphism of the
object  $F^{X'}_{\Theta}({\cal H}_2)$, then the following diagram
takes place:
$$
 \xymatrix{
T  \ar@{<->}[r]^{\alpha_{\varphi}^{-1}}
\ar[d]_{id_{F^{X'}_{\Theta}({\cal H}_2)}}  & \alpha_{\varphi}^{-1}
T \ar[d]^{{\cal F}(id_{F^{X'}_{\Theta}({\cal H}_2)})} \\
T  \ar@{<->}[r]^{\alpha_{\varphi}^{-1}} & \alpha_{\varphi}^{-1} T,
 }
$$
where $T\in F^{X'}_{\Theta}({\cal H}_2)$, $\alpha_{\varphi}^{-1} T
\in F^{X}_{\Theta}({\cal H}_1)$
%
Thus, ${\cal F}'(id_{F^{X'}_{\Theta}({\cal H}_2)})=id_{{\cal
F}'({F^{X'}_{\Theta}({\cal H}_2)})}=id_{{F^{X}_{\Theta}({\cal
H}_1)}}$.

Take two morphisms in $F_{\Theta}({\cal H}_2)$:
$$
[s_{*}^{1}]_{{\cal H}_2}:F_{\Theta}^{X'}({\cal H}_2)\to
F_{\Theta}^{Y'}({\cal H}_2),
$$
such that
$$
[s_{*}^{1}]_{{\cal H}_2}: T_1 \to T_2,
$$
and
$$
[s_{*}^{2}]_{{\cal H}_2}:F_{\Theta}^{Y'}({\cal H}_2)\to
F_{\Theta}^{Z'}({\cal H}_2),
$$
such that
$$
[s_{*}^{2}]_{{\cal H}_2}: T_2 \to T_3,
$$
where $T_1\in F_{\Theta}^{X'}({\cal H}_2)$, $T_2 \in
F_{\Theta}^{Y'}({\cal H}_2)$, $T_3\in F_{\Theta}^{Z'}({\cal
H}_2)$.

We will check that
$$
{\cal F}'([s_{*}^{2}]_{{\cal H}_2}\circ [s_{*}^{1}]_{{\cal
H}_2}])= {\cal F}'([s_{*}^{2}]_{{\cal H}_2}) \circ {\cal
F}'([s_{*}^{1}]_{{\cal H}_2}]).
$$

Indeed,
$$
[s_{*}^{2}]_{{\cal H}_2}\circ [s_{*}^{1}]_{{\cal H}_2}: T_1\to T_3
$$
and
\begin{equation}\label{eq:F'(s2s1)}
 {\cal F}'([s_{*}^{2}]_{{\cal H}_2}\circ [s_{*}^{1}]_{{\cal H}_2}):
\alpha_{\varphi}^{-1}T_1 \to \alpha_{\varphi}^{-1}T_3.
\end{equation}
From the other hand,
$$
{\cal F }'([s_{*}^{1}]_{{\cal H}_2}): \alpha_{\varphi}^{-1} T_1
\to \alpha_{\varphi}^{-1} T_2
$$
and
$$
{\cal F}'([s_{*}^{2}]_{{\cal H}_2}): \alpha_{\varphi}^{-1} T_2 \to
\alpha_{\varphi}^{-1} T_3.
$$
Thus, the composition of functors ${\cal F}'([s_{*}^{1}]_{{\cal
H}_2})$ and ${\cal F}'([s_{*}^{2}]_{{\cal H}_2})$ works as
follows:
\begin{equation}\label{eq:F'(s2)F'(s1)}
{\cal F}'([s_{*}^{2}]_{{\cal H}_1})\circ {\cal
F}'([s_{*}^{1}]_{{\cal H}_2}): \alpha_{\varphi}^{-1} T_1 \to
\alpha_{\varphi}^{-1} T_3.
\end{equation}
Summarizing equations~(\ref{eq:F'(s2s1)}) and
(\ref{eq:F'(s2)F'(s1)}), we get 
$$
{\cal F}'([s_{*}^{2}]_{{\cal H}_2}\circ [s_{*}^{1}]_{{\cal
H}_2})={\cal F}'([s_{*}^{2}]_{{\cal H}_2})\circ {\cal
F}'([s_{*}^{1}]_{{\cal H}_2}).
$$

Thus, the correspondence ${\cal F}': F_{\Theta}({\cal H}_2) \to
F_{\Theta}({\cal H}_1)$ is a covariant functor.

Moreover, one can check that ${\cal F}'$ is the inverse functor
for ${\cal F}$ and categories $F_{\Theta}({{\cal H}_1})$ and
$F_{\Theta}({{\cal H}_2})$ are isomorphic. Theorem is proved.
\end{proof}

The next theorem is the main result concerning  logically
automorphically equivalent knowledge bases.

\begin{theorem}\label{th:Aut-Equiv-KB}
Logically automorphically equivalent knowledge bases   $KB({\cal
H}_1)$ and  $KB({\cal H}_2)$ are informationally equivalent.
\end{theorem}

\begin{proof}
Remind that knowledge base $KB({\cal H}_1)$ and $KB({\cal H}_2)$
are informationally equivalent, if the categories of knowledge
description $F_\Theta({\cal H}_1)$ and $F_\Theta({\cal H}_2)$ are
isomorphic (see Definition~\ref{def:InformEquivKB}).

According to Theorem~\ref{th:Aut-Equiv-Model}, logically
automorphical equivalence of models  ${\cal H}_1$ and ${\cal H}_2$
implies isomorphism of categories of knowledge description
$F_\Theta({\cal H}_1)$ and $F_\Theta({\cal H}_2)$. In turn, this
means that the knowledge base $KB({\cal H}_1)$ and $KB({\cal
H}_2)$ are informationally equivalent.
\end{proof}


\end{document}